\documentclass[a4paper,UKenglish]{lipics-v2018}

\usepackage{microtype}

\newcommand{\eps}{\varepsilon}
\usepackage{floatrow}
\usepackage{float}
\usepackage{amsmath}
\usepackage{amsthm}
\usepackage{amssymb}
\usepackage{amsfonts}
\usepackage{color}
\usepackage{algorithm}
\usepackage{caption}
\usepackage{MnSymbol}
\usepackage{multirow}
\usepackage{hhline}

\def\polylog{\operatorname{polylog}}

\title{Improved Time and Space Bounds for Dynamic Range Mode}

\author{Hicham El-Zein}{Cheriton School of Computer Science, University of Waterloo, Canada}{helzein@uwaterloo.ca}{}{}
\author{Meng He}{Faculty of Computer Science, Dalhousie University, Canada}{mhe@cs.dal.ca}{}{}
\author{J. Ian Munro}{Cheriton School of Computer Science, University of Waterloo, Canada}{imunro@uwaterloo.ca}{}{}
\author{Bryce Sandlund}{Cheriton School of Computer Science, University of Waterloo, Canada}{bcsandlund@uwaterloo.ca}{}{}

\authorrunning{H. El-Zein, M. He, I. Munro, and B. Sandlund} 
\Copyright{Hicham El-Zein, Meng He, J. Ian Munro, Bryce Sandlund}
\subjclass{}

\category{}
\relatedversion{}
\supplement{}
\funding{}

\EventEditors{Yossi Azar, Hannah Bast, and Grzegorz Herman}
\EventNoEds{3}
\EventLongTitle{26th Annual European Symposium on Algorithms (ESA 2018)}
\EventShortTitle{ESA 2018}
\EventAcronym{ESA}
\EventYear{2018}
\EventDate{August 20--22, 2018}
\EventLocation{Helsinki, Finland}
\EventLogo{}
\SeriesVolume{112}
\ArticleNo{25} 
\nolinenumbers 

\begin{document}

\maketitle

\begin{abstract}
Given an array A of $n$ elements, 
we wish to support queries for the most frequent and least frequent element in a subrange $[l, r]$ of $A$. 
We also wish to support updates that change a particular element at index $i$ or insert/ delete an element at index $i$. 
For the range mode problem, our data structure supports all operations in $O(n^{2/3})$ deterministic time using only $O(n)$ space. 
This improves two results by Chan et al. \cite{C14}: 
a linear space data structure supporting update and query operations in $\tilde{O}(n^{3/4})$ time and an $O(n^{4/3})$ space data structure supporting update and query operations in $\tilde{O}(n^{2/3})$ time. 
For the range least frequent problem, we address two variations. 
In the first, we are allowed to answer with an element of $A$ that may not appear in the query range, and in the second, the returned element must be present in the query range. 
For the first variation, 
we develop a data structure that supports queries in $\tilde{O}(n^{2/3})$ time, 
updates in $O(n^{2/3})$ time, 
and occupies $O(n)$ space. 
For the second variation, we develop a Monte Carlo data structure that supports queries in $O(n^{2/3})$ time, updates in $\tilde{O}(n^{2/3})$ time, 
and occupies $\tilde{O}(n)$ space, but requires that updates are made independently of the results of previous queries. 
The Monte Carlo data structure is also capable of answering $k$-frequency queries; that is, the problem of finding an element of given frequency in the specified query range. 
Previously, no dynamic data structures were known for least frequent element or $k$-frequency queries.
\end{abstract}

\section{Introduction}
The mode of a sample is a fundamental data statistic along with median and mean. Given an ordered sequence, the range mode of interval $[l, r]$ is the mode of the subsequence from index $l$ to $r$. 
Building a data structure to efficiently compute range modes allows data analysis to be conducted over any window of one-dimensional data. 
Techniques to answer such queries are relevant to the design of database systems.

The range least frequent problem can be seen as a low-frequency variant of range mode. 
Instead of searching for the most frequent element in an interval of the sequence, 
we query for an element that occurs the fewest number of times. 
We may either restrict our attention to elements that occur at least once in the query range or allow the answer to be an element that occurs zero times in the interval but is present elsewhere in the sequence.

A third range frequency query we consider in this paper is the problem of identifying an element with given frequency $k$ in the specified query interval. 
This problem has been called the range $k$-frequency problem.

Both range mode query and range least frequent query have theoretical connections to matrix multiplication. 
In particular, the ability to answer $n$ range mode queries on an array of size $O(n)$ in faster than $O(n^{\omega/2})$ time, where $\omega$ is the constant in the exponent of the running time of matrix multiplication, 
would imply a faster algorithm for boolean matrix multiplication \cite{C14,C15}. Upon closer examination, this lower bound also applies to the range $k$-frequency problem.

The range mode, range least frequent, and range $k$-frequency problems are part of a set of questions one can ask on a subrange of a sequence. Other queries in this area include:
\begin{itemize}
  \item Sum: Return the sum of elements in the query range.
  \item Min/ Max: Return the minimum/ maximum in the query range.
  \item Median: Determine the median element in the query range.
  \item Majority: Return the element that occurs more than 1/2 the time, if such an element exists.
\end{itemize}
Note that the mean of a range can be reduced to the range sum problem. Table \ref{stateofart} gives an overview of known upper bounds on related static and dynamic range query data structures.

\begin{table}[!ht]
\begin{center}
\begin{tabular}{|c|c|c|c|c|}
\hline
Query Type & Query Time & Update Time & Space & Citation\\
\hline
\multirow{2}{*}{Sum} & $O(1)$ & - & $O(n)$ & trivial\\
\hhline{~----}
& $O(\lg n)$ & $O(\lg n)$ & $O(n)$ & \cite{mih10}\\
\hline
\multirow{2}{*}{Min/ Max} & $O(1)$ & - & $O(n)$ & \cite{GBT84}\\
\hhline{~----}
 & $O(\lg n / \lg\lg n)$ & $O(\lg^{1 / 2 + \eps} n)$ &
$O(n)$ & \cite{chan17}\\
\hline
\multirow{2}{*}{Median} & $O(\lg n / \lg\lg n)$ & - & $O(n)$ & \cite{bro11b}\\
\hhline{~----}
& $O((\lg n / \lg\lg n)^2)$ & $O((\lg n / \lg\lg n)^2)$ &
$O(n)$ & \cite{HMN11}\\
\hline
\multirow{2}{*}{Majority} & $O(1)$ & - & $O(n)$ & \cite{elm11}\\
\hhline{~----}
 & $O(\lg n / \lg\lg n)$ & $O(\lg n)$ &
$O(n)$ & \cite{elm11,gag17}\\
\hline
\multirow{6}{*}{Mode} & $O(n^\eps \lg\lg n)$ & - & $O(n^{2-2\eps})$ & \cite{KMS03}\\
\hhline{~----}
& $O(1)$ & - & $O(n^2 \lg \lg n/ \lg n)$ & \cite{KMS03}\\
\hhline{~----}
& $O(\sqrt{n/\lg n})$ & - & $O(n)$ & \cite{C14}\\
\hhline{~----}
& $O(n^{3/4}\lg n/\lg\lg n)$ & $O(n^{3/4} \lg \lg n)$ & $O(n)$ & \cite{C14}\\
\hhline{~----}
& $O(n^{2/3}\lg n/\lg \lg n)$ & $O(n^{2/3}\lg n/\lg \lg n)$ & $O(n^{4/3})$ & \cite{C14}\\
\hhline{~----}
& $O(n^{2/3})$ & $O(n^{2/3})$ & $O(n)$ & new\\
\hline
\multirow{2}{*}{Least Frequent} & $O(\sqrt{n})$ & - & $O(n)$ & \cite{C15}\\
\hhline{~----}
& $O(n^{2/3} \lg n\lg\lg n)$ & $O(n^{2/3})$ & $O(n)$ & new\\
\hline
$k$-Frequency & $O(n^{2/3})$ & $O(n^{2/3} \lg n)$ & $O(n)$ & new\\
\hline
\end{tabular}
\end{center}
\caption{Known static and dynamic range query upper bounds.}
\label{stateofart}
\end{table}

\subsection{Our Results}

We improve the results of Chan et al. \cite{C14} by giving a dynamic range mode data structure that takes $O(n)$ space and supports updates and queries in $O(n^{2/3})$ time. This improves the query/ update time of their linear space data structure by a polynomial factor and additionally improves the query/ update time of their $O(n^{4/3})$ space data structure by a $O(\log n / \log \log n)$ factor. We also include in our update procedures the ability to insert or delete elements in the middle of the array, an operation not addressed in previous dynamic range mode data structures.

Our improvements are based on the observation that knowing how many of a type of element occur in an interval can be as valuable as knowing the elements themselves. Specifically, instead of storing the frequency counts of elements per span, we store the number of elements with a particular frequency count per span. This information can be dynamically maintained, and uses $O(\log n)$ bits per frequency per span, rather than $O(\log n)$ bits per unique element per span.

Our technique is general enough to also apply to the range least frequent problem in a dynamic setting. To our knowledge, this is the first data structure to do so. In the version of the problem where we allow an answer to not occur in the specified query interval, our data structure supports queries in $O(n^{2/3} \log n \log \log n)$ time, updates in $O(n^{2/3})$ time, and occupies $O(n)$ space. In the version where the least frequent element must be present in the query interval, we develop a Monte Carlo data structure that supports queries in $O(n^{2/3})$ time, updates in $O(n^{2/3} \log n)$ time, and occupies $O(n \log^2 n)$ space. This data structure is correct with high probability for any polynomial sequence of updates and queries, with the restriction that updates are made independently of the results of previous queries. Notably, if the set of queries and updates is fixed in advance or given to the algorithm all at once, this property holds.

Furthermore, our Monte Carlo data structure is powerful enough to apply to the dynamic range $k$-frequency problem, also supporting queries in $O(n^{2/3})$ time, updates in $O(n^{2/3} \log n)$ time, and occupying $O(n \log^2 n)$ space. This data structure can be augmented to count the number of elements below, above, or at a given frequency, supporting both queries and updates in $O(n^{2/3})$ time and using $O(n)$ space, without the need for an independence assumption.

We organize our results as follows. 
In Section \ref{previousw}, we review previous work on static and dynamic range mode and least frequent element queries. 
In Section \ref{prelim}, we briefly give some notation that will be used for the rest of the paper. 
In Section \ref{setup}, we give the basic setup of the $O(n)$ space data structure. 
Section \ref{rmqS} describes how to answer a range mode query in $O(n^{2/3})$ time. 
Section \ref{updateS} explains how to support updates to our base data structures in $O(n^{2/3})$ time. 
In Section \ref{leastfreq}, we discuss how to answer range least frequent queries in $\tilde{O}(n^{2/3})$ time. 
Section \ref{givenfreq} describes how to find an element of a given frequency in a specified range in $O(n^{2/3})$ time. 
Here we also mention additional frequency operations our data structure can support. 
Due to space limitations, we will not describe how the data structure can be made to support insertion and deletion of elements in the conference version of this paper.

\section{Previous Work}
\label{previousw}

\subsection{Static Range Mode Query}

The static range mode query problem was first studied by Krizanc et al. \cite{KMS03}. Their focus is primarily on subquadratic solutions with fast queries, achieving $O(n^{2-2\eps})$ space and $O(n^\eps \log n)$ query time, with $0 < \eps \leq 1/2$, and $O(n^2 \log \log n / \log n)$ space and $O(1)$ query time. If we set $\eps = 1/2$ with the first approach, this gives a linear space static range mode data structure with query time $O(\sqrt{n} \log n)$. By substituting an $O(\log \log n)$ data structure for predecessor search, such as van Emde Boas trees, the query time can immediately be improved to $O(\sqrt{n} \log \log n)$.

Chan et al. \cite{C14} focus on linear space solutions to static range mode. They achieve a clever array-based solution with $O(n)$ space and $O(\sqrt{n})$ query time. By using bit-packing tricks and more advanced data structures, they reduce the query time to $O(\sqrt{n/\log n})$.

As with many range query data structures, the range mode problem has also been studied in an approximate setting \cite{GJLT10,BKMT05}.

\subsection{Static Range Least Frequent Query}

The range least frequent problem was first studied by Chan et al. \cite{C15}. They again focus on linear-space solutions, this time achieving $O(n)$ space and $O(\sqrt{n})$ time query. In their paper, they focus on the version of range least frequent element where the element must occur in the query range.

\subsection{Dynamic Range Mode}

Chan et al. \cite{C14} also study the dynamic range mode problem. They give a solution tradeoff that at linear space, achieves $O(n^{3/4} \log n / \log \log n)$ worst-case time range mode query and $O(n^{3/4} \log \log n)$ amortized expected time update. At minimal update/ query time, this tradeoff gives $O(n^{4/3})$ space and $O(n^{2/3} \log n / \log \log n)$ worst-case time range mode query and amortized expected time update.

\subsection{Lower Bounds}

Both \cite{C14} and \cite{C15} give conditional lower bounds for range mode and range least frequent problems, respectively. Chan et al. \cite{C14} reduces multiplication of two $\sqrt{n} \times \sqrt{n}$ boolean matrices to $n$ range mode queries in an array of size $O(n)$. This indicates that with current knowledge, preprocessing an $O(n)$-sized range mode query data structure and answering $n$ range mode queries cannot be done in better than $O(n^{\omega/2})$ time, where $\omega$ is the constant in the exponent of the running time of matrix multiplication. With $\omega = 2.3727$ \cite{W12}, this implies with current knowledge that a range mode data structure must either have preprocessing time at least $\Omega(n^{1.18635})$ or query time at least $\Omega(n^{0.181635})$. Since we may choose to update an array rather than initializing it, this lower bound also indicates that a dynamic range mode data structure must have update/ query time at least $\Omega(n^{\omega/2 -1})$.

In \cite{C14}, Chan et al. also give another conditional lower bound for dynamic range mode. They reduce the multiphase problem of P\v{a}tra\c{s}cu \cite{P10} to dynamic range mode. A reduction from 3-SUM (given $n$ integers, find three that sum to zero) is given by P\v{a}tra\c{s}cu \cite{P10} to the multiphase problem. Based on the conjecture that the 3-SUM problem cannot be solved in $O(n^{2-\eps})$ time for any positive constant $\eps$, this chain of reductions implies a dynamic range mode data structure must have polynomial time query or update.

The reduction of $\sqrt{n} \times \sqrt{n}$ boolean matrix multiplication to $n$ $O(n)$-sized range mode queries can be adopted to achieve the same conditional lower bound for the range least frequent problem \cite{C15}. Upon examination, the conditional lower bound also applies to $k$-frequency queries.

An unconditional lower bound also exists in the cell probe model for the range mode and $k$-frequency problem. Any range mode/ $k$-frequency data structure that uses $S$ memory cells of $w$-bit words needs $\Omega(\frac{\log n}{\log (Sw/n)})$ time to answer a query \cite{GJLT10}.

\section{Preliminaries}
\label{prelim}

Before we discuss the technical details of our results, it will be helpful to develop some notation. 

As in \cite{C14,C15}, we will denote the subarray of $A$ from index $i$ to index $j$ as $A[i:j]$ and use array notation $A[i]$ to denote the element of $A$ at index $i$. We will assume zero-based indexing throughout this paper.

Furthermore, for range frequency data structures, the actual type that array $A$ stores is irrelevant; we only care about how many times each element occurs. It will be useful to think of the identity of an element as a color. Therefore, we may say the color $c$ occurs with frequency $f$ in range $A[l:r]$. This is to distinguish that we are not referring to a particular index but rather the identity of multiple indices in the range.

For simplicity, we will assume the number of elements $n$ is a perfect cube; however, the results discussed easily generalize for arbitrary $n$. All $\log$'s in this paper are assumed to be base $2$.

In Appendix \ref{insertdelete} we will discuss operations that insert and delete elements into and from array $A$. 
This effectively changes the indexing of elements after the point of operation. 
It will be helpful to develop an indexing scheme that does not change with the effect of insertion or deletion. 
To differentiate between the interface of the operations and the internals used by the data structure, we will denote with every position of $A$ two values: the \textit{rank} and \textit{index} of the corresponding position. 
The interface of all our operations will be based on rank; however, the internals of our data structure will be based on index. 
This will be further discussed in Appendix \ref{insertdelete} and is only important when insert and delete operations are permitted.

In all our proofs, we analyze space cost in words, that is, $O(\log n)$ collections of bits.

\section{Data Structure Setup}
\label{setup}

The idea of our data structure will be to break up array $A$ into $O(n^{1/3})$ evenly-spaced endpoints, so that there are $O(n^{2/3})$ elements between each endpoint. We will use capital letters $L$ and $R$ when referring to particular endpoints. We will also occasionally refer to the elements between two consecutive endpoints as a segment; therefore, there are $O(n^{1/3})$ segments in $A$. 

Each color that occurs in $A$ will be split into the following two disjoint categories:

\begin{itemize}
  \item \textbf{Frequent Colors}: Any color that appears more than $n^{1/3}$ times in $A$.
  \item \textbf{Infrequent Colors}: Any color that appears at most $n^{1/3}$ times in $A$.
\end{itemize}

Note that there can be at most $n / n^{1/3} = n^{2/3}$ frequent colors in $A$ at any point in time.

Our data structure will need to use dynamic arrays as auxiliary data structures. These dynamic arrays will be a modification of the simple two-level version of the data structure described by Goodrich and Koss \cite{GK99}. We have the following lemma regarding the performance of these dynamic arrays:

\begin{lemma}
\label{darrayL}
There is a dynamic array data structure $D$ that occupies $O(n)$ space and supports:
\begin{enumerate}
  \item $D[i]$: Retrieve/ Set the element at rank $i$, in $O(1)$ time.
  \item $Insert(i, x)$: Insert the element $x$ at rank $i$, in $O(\sqrt{n})$ time.
  \item $Delete(i, x)$: Delete the element $x$ at rank $i$, in $O(\sqrt{n})$ time.
  \item $Rank(ptr)$: Determine the rank of the element pointed to by $ptr$, in $O(1)$ time.
\end{enumerate}
\end{lemma}

We leave the proof of Lemma \ref{darrayL} and discussion of the dynamic array data structure to Appendix \ref{darray}. 
If desired, the data structure can be replaced by a balanced binary search tree, at the cost of additional logarithmic factors in the query times of our data structure.\\

We can now describe the base set of auxiliary data structures used throughout the paper:
\begin{enumerate}
 \item Arrays $B_{L, R}$, for all pairs of endpoints $L, R$, indexed from $0$ to $n^{1/3}$, so that
 \[
 B_{L, R}[i] := \text{The number of infrequent colors with frequency $i$ in $A[L:R]$.}
 \]
 \item Dynamic arrays $D_c$, for every color $c$, so that
 \[
 D_c[i] := \text{The index in $A$ of the $i$th occurrence of color $c$.}
 \]
 \item An array $E$ parallel to $A$ so that
 \[
 E[i] := \text{A pointer to the location in memory of index $i$ in dynamic array $D_{A[i]}$.}
 \]
 \item A binary search tree $F$ of endpoints. At each endpoint $R$, we store
 \[
 F[R] := \text{A binary search tree on frequent colors, giving their frequency in $A[0:R]$.}
 \]
\end{enumerate}

In regards to $B_{L, R}$, we will sometimes refer to the set of elements between endpoints $L$ and $R$ as the span of $L$ and $R$.

We now analyze the space complexity and construction time.

\begin{lemma}
The base data structures take $O(n)$ space.
\end{lemma}

\begin{proof}
The arrays $B_{L, R}$ have size $O(n^{1/3})$ and there are $O(n^{2/3})$ of them, so in total these take $O(n)$ space. Every index of $A$ is present in exactly one of the $D_c$ arrays, and each dynamic array takes linear space. Therefore, in total all $D_c$ arrays take $O(n)$ space. Array $E$ has exactly the same size as $A$ and thus takes $O(n)$ space. The binary search trees in each node of $F$ have size equal to the number of frequent colors, which is at most $n^{2/3}$. Since there are $O(n^{1/3})$ endpoints and thus nodes of $F$, this structure takes $O(n)$ space.
\end{proof}

\begin{lemma}
\label{build}
The base data structures can be initialized in $O(n^{4/3})$ time.
\end{lemma}

\begin{proof}
We can count the number of occurrences of each color in $O(n \log n)$ time and determine for each color whether it is frequent or infrequent. Let $A'$ be the array $A$ without any frequent colors. We can scan $A'$ $n^{1/3}$ times, starting from each endpoint, to build the arrays $B_{L, R}$. This will be done as follows. For each scan, we maintain an array $T$ so that $T[c]$ denotes the number of occurrences of color $c$ found so far. We also maintain the array $B_{L, *}$ which is the array $B$ with endpoint $L$ and a variable right endpoint, that is maintained as elements are scanned. When $A[i] = c$, we check the number of occurrences of $c$ to update $B_{L, *}$ to the correct state. In total, each element scanned results in $O(1)$ operations, until we reach a right endpoint. When we reach a right endpoint, we write $B_{L, *}$ to array $B_{L, R}$, where $R$ denotes the right endpoint just encountered. In this way, for all endpoints $L, R$, we spend $O(n^{1/3})$ time to create the array. Thus the element scan dominates the time complexity, requiring $O(n^{4/3})$ time to create all $B_{L, R}$ arrays.

The dynamic arrays $D_c$ can be built in linear time overall by walking through $A$ and appending indices to the ends of $D_c$ arrays. At this same time $E$ can be built. The BSTs for all endpoints in $F$ can also be built in linear time overall, since there are at most $n^{2/3}$ frequent colors per list, there are $n^{1/3}$ lists in total, and counting each frequent element in each interval can be done at the same time in one scan through $A$.
\end{proof}

Throughout the next few sections, we will use the following Lemma, as in \cite{C14}:

\begin{lemma}
\label{questions}
Let $A[i] = c$. Then, given frequency $f$ and right endpoint $j$, our base data structures may answer the following questions in constant time:
\begin{align}
&\text{Does color $c$ occur at least $f$ times in $A[i:j]$?} \label{q1} \\
&\text{Does color $c$ occur at most $f$ times in $A[i:j]$?} \label{q2} \\
&\text{Does color $c$ occur exactly $f$ times in $A[i:j]$?} \label{q3}
\end{align}
\end{lemma}
\begin{proof}
We call $rank(E[i])$ on $D_c$ to get the rank $r$ of $i$ in $D_c$. Since array $D_c$ stores the indices of every occurrence of color $c$ in $A$, the values of $D_c[r+f]$ and $D_c[r+f-1]$ determine the answers to the above questions.
\end{proof}
A similar strategy can be used to answer the above questions given an index $j$ and left endpoint $i$.

\section{Range Mode Query}
\label{rmqS}
The range mode query will make use of the following lemma, originating from \cite{KMS03} and also used by \cite{C14}:

\begin{lemma}[Krizanc et al. \cite{KMS03}]
\label{split}
Let $A_1$ and $A_2$ be any multisets. If $c$ is a mode of $A_1 \cup A_2$ and $c \notin A_1$, then $c$ is a mode of $A_2$.
\end{lemma}

The query algorithm can be summarized as follows:

\begin{algorithm}[H]
\caption{Range Mode Query in $A[l:r]$}
\label{modeq}
\label{freq}
\begin{flushleft}
Let $L$ and $R$ be the first and last endpoints in $[l, r]$, respectively.
\begin{enumerate}
  \item Check the frequency of every frequent color in $A[L:R]$ via BSTs $F[R]$ and $F[L]$. Let $f$ be the highest frequency found so far.
  \item Ask question \eqref{q1} for all colors in $A[l:L-1] \cup A[R+1:r]$ with frequency $f$ and right endpoint $r$/ left endpoint $l$. If \eqref{q1} is answered in the affirmative for color $c$, linearly scan $D_c$ to count the number of occurrences of color $c$ in $[l, r]$, update $f$, and continue.
  \item Find the largest nonzero index of $B_{L, R}$. If this is larger than $f$, update $f$ and do the following:
  \begin{enumerate}
    \item Find the next endpoint $R'$ to the left of $R$.
    \item Check $B_{L, R'}[f]$. If $B_{L, R'}[f] < B_{L, R}[f]$, search $A[R'+1:R]$ for a color that occurs $f$ times in $A[L:R]$, via question \eqref{q3} with left endpoint $L$.
    \item Otherwise, repeat from step (a) with $R \leftarrow R'$.
  \end{enumerate}
  \item Return $f$ and the corresponding color found from either step 1, 2, or 3(b).
\end{enumerate}
\end{flushleft}
\end{algorithm}

\begin{theorem}
Algorithm \ref{modeq} finds the current range mode of $A[l:r]$ in $O(n^{2/3})$ time.
\end{theorem}
\begin{proof}
Endpoints $L$ and $R$ can be found from $[l, r]$ by appropriate floors and ceilings in constant time.

In step 1, we can iterate through $F[R]$ and $F[L]$ in $O(n^{2/3})$ time, determining frequency counts for all frequent elements.

In step 2, answering question \eqref{q1} takes $O(1)$ time per element via Lemma \ref{questions}. 
Note that we use left endpoint $l$ for elements in range $A[R+1:r]$ and right endpoint $r$ for elements in range $A[l:L-1]$.  
Although we do not check \eqref{q1} for the full range $[l,r]$, 
we will check the first/ last occurrence of any color in $A[l:L-1] \cup A[R+1:r]$, thereby effectively checking for all of $[l,r]$. 
When \eqref{q1} is answered in the affirmative at index $i$, 
we linearly scan $D_c$, starting at $D_c[i + f + 1]$ ($D_c[i-f-1]$ if $i \in [R+1:r]$) to determine a new highest frequency.
%
%
Let us determine the cost of these linear scans. 
Let $c$ be the most frequent color found from steps 1 and 2. If $c$ is an infrequent color, 
it cannot occur more than $n^{1/3}$ times, so the total cost of the linear scans is no more than $O(n^{1/3})$. 
If it is a frequent color, its frequency cannot have increased by more than $O(n^{2/3})$, since its frequency was checked in step 1 and only $O(n^{2/3})$ elements exist in $A[l : L-1] \cup A[R+1 : r]$. 
Thus the total cost of the linear scans is no more than $O(n^{2/3})$. 
In either case, step 2 takes $O(n^{2/3})$ time.

For step 3, finding the largest nonzero index of $B_{L, R}$ takes $O(n^{1/3})$ time. If this is larger than the frequencies found in steps 1 or 2, we execute steps 3(a) - 3(c). We can only repeat the steps at most $O(n^{1/3})$ times. The condition $B_{L, R'}[f] < B_{L, R}[f]$ will happen for one of these iterations, since eventually $R' = L$ in which case there are no elements accounted for in the interval. When $B_{L, R'}[f] < B_{L, R}[f]$, this implies a color with frequency $f$ in range $[L, R]$ appears somewhere in $A[R' + 1 : R]$. There are only $O(n^{2/3})$ elements in $A[R' + 1 : R]$. Checking if color $c$ occurs exactly $f$ times in $A[L:R]$ can be done by asking question \eqref{q3}.

For the correctness of the value returned in step 4, note that the mode of $A[l:r]$ is either an element in $A[l:L-1] \cup A[R+1:r]$ or the mode of $A[L:R]$, by Lemma \ref{split}. The frequency of all colors in $A[l:L-1] \cup A[R+1:r]$ is checked. Further, the frequency of infrequent and frequent colors for interval $A[L:R]$ is also checked in steps 1 and 3, respectively. Therefore the color and frequency returned in step 4 must be the mode of $A[l:r]$.
Putting it together, we see Algorithm \ref{modeq} is correct and takes $O(n^{2/3})$ time.
\end{proof}

\section{Update Operation}
\label{updateS}

The update operation will require us to keep the base data structures up to date. Given update $A[i] \leftarrow c$, there are two similar procedures that must occur: adjusting the data structures for the removal of current color $A[i]$ and adjusting the data structures for the addition of color $c$ at index $i$. 

The following algorithm can be used for both procedures. We note that if either the color removed is the last occurrence of its type or the color added is a new color, the list $D_c$ will have to also be constructed/ deleted and $B_{L, R}$ should be modified to only reflect colors present in $A$. For simplicity we omit these details from Algorithm \ref{update}.

\begin{algorithm}[H]
\caption{Update Base Data Structures for Addition/ Removal of Color $c$ at Index $i$.}
\label{update}
\begin{flushleft}
\begin{enumerate}
  \item If color $c$ is infrequent prior to this operation, count how many times $c$ occurs in each span via $D_c$, decrementing the corresponding index of each $B$ array.
  \item Adjust $D_c$ by adding/ removing $i$. If this is an add operation, set $E[i]$ to the memory location of $i$ in $D_c$.
  \item If color $c$ remains or becomes infrequent after step 2, again count how many times $c$ occurs in each span via $D_c$, incrementing the corresponding index of each $B$ array.
  \item If color $c$ became infrequent in step 2, delete its entry in all nodes of $F$. If color $c$ became frequent after step 2, add its frequencies to $F$. If color $c$ remained frequent after step 2, increment the frequencies of all prefixes including index $i$ in $F$.
\end{enumerate}
\end{flushleft}
\end{algorithm}

\begin{theorem}
Algorithm \ref{update} updates the base data structures for addition/ removal of color $c$ at index $i$ in $O(n^{2/3})$ time.
\end{theorem}
\begin{proof}
If $c$ is an infrequent color prior to the update, step 1 removes its contribution to all $B$ arrays. Counting the frequency of color $c$ in each interval can be done via $O(n^{1/3})$ searches through $D_c$, each taking  $O(n^{1/3})$ time. We first start at the beginning, finding its frequency for all right endpoints with left endpoint fixed, then move the left endpoint to the next endpoint and repeat, etc. Step 1 in total takes $O(n^{2/3})$ time.

Adding or removing $i$ from $D_c$ takes $O(\sqrt{n})$ time, as given in Lemma \ref{darrayL}. Adjusting $E[i]$ can be done during this operation, so in total step 2 takes $O(\sqrt{n})$ time.
The analysis of step 3 is identical to step 1.
In step 4, adding, deleting, or modifying the entry of color $c$ in all/ some nodes of $F$ takes $O(n^{1/3} \log n)$ time, since each node stores the list of frequent colors and their frequency counts in a BST. If color $c$ just became frequent, it only occurs $O(n^{1/3})$ times in $A$, and thus counting its frequency from the beginning to each endpoint can be done in $O(n^{1/3})$ time. In total, step 4 takes at most $O(n^{1/3} \log n)$ time.

After the completion of Algorithm \ref{update}, $B$ arrays, dynamic array $D_c$, array $E$, and binary search tree $F$ has been updated to reflect the current state of array $A$. Since all steps execute in no more than $O(n^{2/3})$ time, Algorithm \ref{update} is correct and runs in $O(n^{2/3})$ time.
\end{proof}

\section{Range Least Frequent Query, Allowing Zero}
\label{leastfreq}

To answer range least frequent queries, we require the use of one more auxiliary data structure:
an array $C_{L, R}$ for all pairs of endpoints $L, R$, indexed from $1$ to $n^{1/3}$. 
At index $C_{L, R}[i]$ we store a list of all colors that occur $i$ times in $A$ such that the smallest span enclosing all occurrences is $[L, R]$.

Since each infrequent color is represented exactly once in the $C_{L, R}$ lists and there are $O(n^{2/3} \cdot n^{1/3}) = O(n)$ total indices present, this data structure takes linear space. It can also be initialized in linear time, and we can modify the update procedure of Algorithm \ref{update} to update $C_{L, R}$ at the same time as $B_{L, R}$ for no additional time cost.

It is additionally worth noting that since the Range Least Frequent Query will require $\tilde{O}(n^{2/3})$ time, the use of dynamic arrays for data structures $D_c$ is not necessary for this section. Instead, we can use binary search trees, augmented to support lookup by index in $O(\log n)$ time, or Dietz' data structure \cite{D89}. For the best time complexity, we will use augmented binary search trees to count occurrences of colors in a specific range and an augmented dynamic linear-space van Emde Boas 
 tree to count occurrences of colors between endpoints. The van Emde Boas tree stores a single node for any endpoint $R$ that a color appears in, which keeps the number of occurrences of that color from the beginning to endpoint $R$. The van Emde Boas tree can be updated in $O(n^{1/3} \log \log n)$ time upon insertion or removal and via predecessor/ successor queries, can support counting the number of occurrences of any color between any two endpoints in $O(\log \log n)$ time.


\begin{figure}
\begin{algorithm}[H]
\caption{Range Least Frequent Query in $A[l:r]$, Allowing Zero}
\label{freq}
\begin{flushleft}
Let $L$ and $R$ be the first and last endpoints in $[l, r]$, respectively.
\begin{enumerate}
  \item Check the frequency of every frequent color in $A[l:r]$ via $D_c$. Let $f$ be the lowest frequency found so far.
  \item Find the set of all infrequent colors that occur in $A[l:L-1] \cup A[R+1:r]$; call it $U$. Count the frequencies of all colors of $U$ in range $A[l:r]$ and update $f$ with the lowest frequency found so far.
  \item Compute $B'_{L, R}$, the array $B_{L, R}$ updated to erase the contribution of all colors in $U$.
  \item If the smallest positive-valued index of $B'_{L, R}$ is less than $f$, update $f$ and check if any list $C_{L', R'}[f]$, $[L', R'] \subseteq [L, R]$, is non-empty. If so, return a color from the appropriate list. Otherwise, binary search from $R$ to the last endpoint of $A$ in the following way:
  \begin{enumerate}
    \item Let $R'$ be a value in the middle of the search range. If $B'_{L, R'}[f] < B'_{L, R}[f]$, let $R'$ be the new upper bound; otherwise, continue the search in the half of the range after $R'$.
    \item When the search range is two consecutive endpoints, we search the range for a color that occurs $f$ times in $A[L:R]$ and return its identity.
    \item If the condition in a. is never satisfied, we must repeat a binary search on the other side, with an initial search range of the beginning of $A$ to $L$.
  \end{enumerate}
\end{enumerate}
\end{flushleft}
\end{algorithm}
\end{figure}

\begin{theorem}
Algorithm \ref{freq} finds the least frequent element of $A[l:r]$, allowing zero, in $O(n^{2/3}\log n \log \log n)$ time.
\end{theorem}

\begin{proof}
We can find a list of frequent colors in any node of the $F$ BST. Using an augmented binary search tree for $D_c$, step 1 can then be done in $O(n^{2/3} \log n)$ time. In step 2, since there are $O(n^{2/3})$ elements in $A[l:L-1] \cup A[R+1:r]$, we can find color set $U$ and complete step 2 similarly to step 1 in $O(n^{2/3} \log n)$ time.  Step 3 requires counting the frequency of each color of $U$ in range $A[L:R]$ and decrementing the corresponding index to make $B'_{L, R}$; thus it can also be done in $O(n^{2/3} \log \log n)$ time using the augmented van Emde Boas tree.

In step 4 we are looking for a least frequent element of $A[L:R]$ that does not occur in $A[l:L-1] \cup A[R+1:r]$. All colors of $A[l:L-1] \cup A[R+1:r]$ are represented in set $U$, found in step 3. We effectively erase the contribution of colors of $U$ via computing $B'_{L, R}$; therefore, we can proceed as if colors of $U$ do not exist in $A$. We will refer to $A'$ as array $A$ without any colors of $U$.

If the least frequent color in $A'[L:R]$ does not exist elsewhere in $A$, then its identity will be stored in a list $C_{L', R'}[f]$, $[L', R'] \subseteq [L, R]$. Note colors of $U$ need not be special-cased for this lookup, since by appearing in $A[l:L-1] \cup A[R+1:r]$, they will not be present in any of the searched lists. Otherwise, we know the least frequent color in $A'[L:R]$ must occur somewhere else in $A'$.

Now, amongst all colors in $A'$, we know frequency $f$ is minimal in range $A'[L:R]$. Therefore if we increase the range to $A'[L:R']$, the frequency of colors can only increase. For this property to hold, we must allow $f = 0$. 


In each iteration, the smallest positive-valued index of $B'_{L, R'}$ will be $f$ or greater and $B'_{L, R'}[f]$ will be no more than $B'_{L, R}[f]$. If it is less, we know one of the colors that occurred $f$ times in $A'[L:R]$ now occurs more than $f$ times in $A'[L:R']$. Therefore we may find it in the half of the search range before $R'$. If it is the same, we know none of the colors that occurred $f$ times in $A'[L:R]$ appear in the half of the search range before $R'$. Either way we decrease the search range by a factor of $2$. When the range represents two consecutive endpoints, we can search it for a color that occurs $f$ times in $A[l:r]$ in $O(n^{2/3} \log n)$ time.

However, if $R'$ is the end of the array and $B'_{L, R'}[f] = B'_{L, R}[f]$, then none of the colors that appear $f$ times in $A'[L:R]$ appear to the right of $R$ in $A'$. In this case, we can repeat the same binary search on the other side, decreasing a left endpoint $L'$ and checking the same condition. Since the least frequent color in $A'[L:R]$ must occur elsewhere in $A$, as it was not present in any of the lists $C_{L', R'}[f]$, the search on this side must identify a color that appears $f$ times in $A'[L:R]$.

The time complexity of step 4 can be analyzed as follows. Checking the lists $C_{L', R'}$ takes $O(n^{2/3})$ time, since there can be $O(n^{2/3})$ endpoints $[L', R'] \subseteq [L, R]$. The binary search is on endpoints, of which there are $O(n^{1/3})$. Thus, there are $O(\log(n^{1/3})) = O(\log n)$ iterations of the binary search, and in each iteration we must compute $B'_{L, R'}$ or $B'_{L', R}$. This computation takes $O(n^{2/3} \log \log n)$ time as in step 3. Therefore the binary search process takes $O(n^{2/3}\log n \log \log n)$ time. In total, step 4 takes $O(n^{2/3}\log n \log \log n)$ time.

For correctness, the least frequent element in $A[l:r]$ is either a frequent or infrequent color. If it is a frequent color, it is identified in step 1. If it is an infrequent color, we have two cases. Either the color occurs in $A[l:L-1] \cup A[R+1:r]$, and thus set $U$, or it does not occur in set $U$. Step 2 accounts for all infrequent colors in set $U$. Steps 3 and 4 account for the last case. By the above, these steps find an infrequent color in $A[L:R]$ that does not occur in $A[l:L-1] \cup A[R+1:r]$ in $O(n^{2/3}\log n \log \log n)$ time. Thus, Algorithm \ref{freq} is correct and finds the least frequent element of $A[l:r]$, allowing zero, in $O(n^{2/3}\log n \log \log n)$ time.
\end{proof}

\section{Range $k$-Frequency Query}
\label{givenfreq}

The previous two sections make use of a monotonicity property to find a color of given frequency in a range: for range mode, we know the frequency of the most frequent element can only decrease if the query range is decreased; furthermore, for range least frequent, we know the frequency of the least frequent element can only increase if the query range is increased. For this monotonicity condition to hold for least frequent elements, we must allow answering with an element of frequency zero. To force our answer to be an element that occurs in the query range, we must use an additional data structure that allows retrieval of colors by frequency in the $B_{L, R}$ arrays. To achieve $\tilde{O}(n)$ space, we cannot afford to store a list of colors at each index $B_{L, R}[i]$, and storing a single color at each index will run into issues during updates.
Instead, we will use the following data structure which is similar to randomized data structures in the dynamic streaming literature \cite{fra08}:

\begin{lemma}
\label{montecarloL}
There is a Monte Carlo data structure that occupies expected $O(\log^2 n)$ space and supports:
\begin{enumerate}
  \item $Insert(x)$: Insert element $x$ into the collection, in $O(\log n)$ expected time.
  \item $Delete(x)$: Delete element $x$ from the collection, in $O(\log n)$ expected time.
  \item $Retrieve()$: Return an element in the collection, in $O(1)$ expected time.
\end{enumerate}
The data structure requires the $Delete(x)$ operation is executed independently of the results of $Retrieve()$.
\end{lemma}
We leave the proof of Lemma \ref{montecarloL} and discussion of the Monte Carlo data structure to the Appendix \ref{montecarlo}.
With it, we can answer the general problem of finding an element of given frequency in a query range. The additional auxilliary data structures needed will be as follows:
\begin{enumerate}
  \item An array $G_{L, R}$ parallel to $B_{L, R}$. At index $G_{L, R}[i]$, we store the number of infrequent colors with frequency $i$ in $A[L:R]$, excluding colors that appear in segments immediately left of $L$ or right of $R$.
  \item An array $H_{L, R}$, parallel to $G_{L, R}$, so that at index $H_{L, R}[i]$, we store a collection of colors counted in $G_{L, R}[i]$ in the data structure of Lemma \ref{montecarloL}.
\end{enumerate}
The array $G_{L, R}$ is similar to the item (ii) stored in table $D$ of \cite{C15}. During preprocessing, $G$ arrays can be built similarly to $B$; however, when we fix left endpoint, we check all colors that occur in the segment to the left. We avoid counting such colors. Similarly, before we finalize the count for $G_{L, R}$, we move the right endpoint out as if to count the next range, keeping track of colors encountered. We subtract the frequency of such colors in $G_{L, R}$. Whenever we add to/ subtract from $G_{L, R}$, we can insert or delete the color from $H_{L, R}$.

As explained above, our space cost now becomes $\tilde{O}(n)$. 
Furthermore, our updates to $G$ and $H$ can be done alongside the update to $B$; however, since each insertion takes $O(\log n)$ time, the update procedure now takes $O(n^{2/3}\log n)$ time.
With this we have:

\begin{algorithm}[H]
\caption{Range $k$-Frequency Query in $A[l:r]$}
\label{rfq}
\begin{flushleft}
Let $L$ and $R$ be the first and last endpoints in $[l, r]$, respectively.
\begin{enumerate}
  \item For color $c$ at index $i$ in the segment left of $L$, check via $D_c$ to see if $i$ is the first index of color $c$ to appear in range $[l, r]$, or, if outside $[l, r]$, the next occurrence of color $c$ lies in range $[l, r]$. If so, ask question \eqref{q3} for the occurrence of color $c$ in $[l, r]$ with frequency $k$ and right endpoint $r$. If answered in the affirmative, return color $c$. Do the same, symmetrically, for colors in the segment right of $R$.
  \item For each frequent color not addressed in step 1, check its frequency in $A[L:R]$ via BSTs $F[R]$ and $F[L]$. If any occur with frequency $k$, return the color.
  \item If no color is found from step 1, check if $G_{L, R}[k] > 0$. If so, return $H_{L, R}[k].Retrieve()$. If not, return that no color has frequency $k$ in range $A[l:r]$.
\end{enumerate}
\end{flushleft}
\end{algorithm}

\begin{theorem}
\label{rfqT}
Algorithm \ref{rfq} returns an element of frequency $k$ in $A[l:r]$ or indicates no such element exists, assuming update operations have been executed independently of results of query operations, in $O(n^{2/3})$ time.
\end{theorem}

\begin{proof}
In step 1, we look at two full segments of $O(n^{2/3})$ total elements. 
For color $c$ at index $i$, if $i \in [l, r]$, we must determine if index $i$ is the first occurrence of color $c$ in $[l, r]$. 
Let $m = D_c.rank(E[i])$. 
Index $i$ is the first occurrence of color $c$ in $[l, r]$ if $D_c[m-1]$ is outside $[l, r]$. 
Similarly, if $i \notin [l, r]$, we can again define $m = D_c.rank(E[i])$, then check if $D_c[m+1]$ is in $[l, r]$. 
In any case, for any color that occurs in segments immediately left of $L$ or right of $R$, 
one of the indices will be the first outside $[l, r]$ or the first within $[l, r]$. 
Thus the frequency of the color will be checked in $A[l:r]$. 
Since we do a constant number of constant time operations for $O(n^{2/3})$ elements, step 1 takes $O(n^{2/3})$ time.

As in Algorithm \ref{modeq}, step 2 takes $O(n^{2/3})$ time. 
Step 3 takes $O(1)$ time. 
In any case, 
in step 1 we check all elements in segments immediately left of $L$ or right of $R$ to see if they occur $k$ times in $A[l:r]$. 
Furthermore, in steps 2 and 3, we check every frequent and infrequent color that occurs in $A[L:R]$ but not in segments immediately left of $L$ or right of $R$, via array $G$. 
Thus we have checked every color if it occurs $k$ times in $A[l:r]$. Since no step takes more than $O(n^{2/3})$ time, this proves Theorem \ref{rfqT}.
\end{proof}

Algorithm \ref{rfq} can be easily modified to return the least frequent element present in the query range with the same time complexity and independence assumption. It can also be modified to count the number of elements above, below, or at a given frequency, as well as only determine the frequency of the least frequent element. Since these queries do not ask for a color, arrays $H_{L, R}$ are not needed. This reduces the space cost to $O(n)$, update cost to $O(n^{2/3})$, and requires no independence assumption.

\clearpage


\appendix
\begin{center}
  \bfseries \sffamily \LARGE Appendix
\end{center}

\section{Modified Dynamic Array Data Structure}
\label{darray}

For the purposes of this paper, 
we will only need the two-level structure of \cite{GK99}. 
It should be possible to generalize the operations discussed for the two-level structure to the multi-level structure, if desired.

Suppose we wish to represent a dynamic array on $n$ elements.
We define a parameter $n_f=\Theta(n)$;
the value of $n_f$ changes as $n$ becomes too large or too small.
Moreover, we define $m$ to be $\lceil \sqrt{n_f} \rceil$. 
We will maintain $\lceil n / m \rceil$ circular queues $Q_i$ of exactly $m$ elements each, 
except possibly for the last queue. 
Every element in each queue has a place in a total order represented by the queues. 
We say all elements of $Q_i$ precede elements of $Q_{i+1}$, 
and within a queue $Q_j$, 
distance from the head of the queue determines the relative order of the elements.
Each circular queue $Q_i$ supports the following operations:
\begin{enumerate}
  \item Insert element $x$ at index $j$ of $Q_i$, then pop the last element off the queue and return it.
  \item Delete the element at index $j$ of $Q_i$, sliding elements after $j$ forward in the queue.
  \item Push element $x$ onto the front/ back of the queue, then pop the first/ last element off the queue and return it.
  \item Return the element at index $j$ of $Q_i$.
  \item Determine the rank of object $x$ in $Q_i$.
\end{enumerate}

Standard array-based implementations of circular queues support operations 1 and 2 in $O(m)$ time and operations 3 and 4 in $O(1)$ time. Operation 5 can be implemented in $O(1)$ time as follows. Underlying the circular queue is an array. For each object $x$, we store its position in that array. The rank of object $x$ is just the position of object $x$ in the array, minus the position of the head of the queue in the array, modulo $m$.

We can now prove Lemma \ref{darrayL} from Section \ref{setup}.

\begin{proof}[Proof of Lemma \ref{darrayL}]
 Since each circular queue has exactly $m$ elements, 
 retrieving/ setting the element at rank $i$ can be done via retrieving/ setting $Q_{i/m}[i \mod m]$. 
 The arithmetic takes constant time, 
 as does retrieving/ setting $Q_{i/m}[i \mod m]$, 
 so this operation takes $O(1)$ time.
 
 Insertion at rank $i$ requires inserting into $Q_{i/m}$ at index $(i \mod m)$. 
 This pops the last element off $Q_{i/m}$, 
 which will then be pushed onto the following queue, 
 and so on. 
 The insertion into $Q_{i/m}$ takes $O(m) = O(\sqrt{n_f}) = O(\sqrt n)$ time. 
 The cascading insert-at-head, pop-from-tail operations take $O(1)$ time each and must be executed no more than $n / m$ times. 
 Thus insertion takes $O(n / m) = O(n / \sqrt{n_f}) = O(\sqrt{n})$ time.
 
 Deletion at rank $i$ is similar to insertion. 
 We delete index $(i \mod m)$ from $Q_{i/m}$. 
 We then remove the head of $Q_{n / m}$, and push it to the previous queue, and so on, until the vacant position at the end of queue $Q_{i/m}$ is filled. 
 Deletion from $Q_{i/m}$ takes $O(m)$ time, and there are no more than $n/m$ cascading operations that each take $O(1)$ time. 
 In total, deletion takes $O(m + n/m) = O(\sqrt{n})$ time.
 
 The rank of a given pointer $ptr$ can be determined by the rank function in the queue in which it resides. 
 Let $r$ be its rank. We can then return $r + im$, where $i$ is the index of the queue in which the object $ptr$ points to resides.
 
The data structure must be resized when $n$ becomes too large or too small with respect to $n_f$. 
We can build the above data structure in $O(n)$ time and therefore we can afford to rebuild from scratch when either condition occurs to achieve $O(\sqrt{n})$ amortized time insertion and deletion. By starting the rebuild at a fast pace when we approach the threshold on $n_f$, the amortized bound can be made worst-case.
\end{proof}




\section{A $\polylog n$ Space Monte Carlo Data Structure}
\label{montecarlo}

The idea of the data structure will be to sample at densities $1/2^i$ for $i$ up to about $\log n$. No matter the size of the data structure, there is a constant probability that the top density has only one element in it. We repeat the data structure with $c \lceil\log n\rceil$ independent copies to achieve a high probability that at any point in time, we may sample one element from the data structure.

To make the $Delete(x)$ operation functional, we need to determine at which densities element $x$ resides, as well as have the guarantee that $x$ currently exists in the collection. We must use the same random bits we used upon insert so we remove $x$ only from the correct densities. We can use the same model as universal families of hash functions to make the insert and delete functions dependent on the element, so behavior is deterministic, but analysis of expected results can be made over the choice of hash functions.

Note that independence of results of $Retrieve()$ and the $Delete(x)$ function is necessary; otherwise, we may repeatedly delete the element returned by $Retrieve()$. If this occurs, the elements are not sampled according to the desired exponential distribution, since elements at top densities have been specifically removed.

Our analysis will be similar to the cutset data structure of \cite{KKM13}. We can now prove Lemma \ref{montecarloL}.

\begin{proof}[Proof of Lemma \ref{montecarloL}]
When we insert $x$ into the data structure, we choose a number $l$ so that $l = 1$ with probability $1/2$, $l = 2$ with probability $1/4$, and so on, so that $l = i$ with probability $1/2^i$. We refer to $l$ as the \textit{level} of element $x$. As explained above, we make the choice of $l$ deterministic to element $x$ so that we can determine the level of $x$ if we delete $x$. With a universal family of hash functions, this can be accomplished by taking the hash of $x$ in binary and using this as the random bit sequence. In this sense, our analysis beyond this point will reason with the expected behavior over the choice of hash functions.

Let $l_{max}$ denote the maximum level currently represented in the data structure. We will maintain a mask for every level $1, \ldots, l_{max}$, as well as a counter at each mask. When we insert $x$ into the data structure, we XOR its value with the mask at its level and increase the counter. If this causes $l_{max}$ to increase, we add the necessary masks and counters to accomplish the task. When we execute $Delete(x)$, we find the level of $x$ and again XOR its value with all masks at or below its level, this time decreasing the appropriate counters and possibly decreasing $l_{max}$. The $Retrieve()$ function works as follows. We check level $l_{max}$ to see if the counter is at $1$. If so, we return the element represented in the mask at level $l_{max}$. If not, we report \texttt{failure}.

Let $n$ be an upper bound on the number of elements the collection will need to represent at any point in time. To achieve a high probability result, we will maintain $c \lceil\log n\rceil$ copies of the above data structure. Insert and delete will be done on each copy independently. The $Retrieve()$ function will only need to be executed until a single copy does not report \texttt{failure}. We analyze the probability of a single copy reporting failure as follows.

Let $k$ be the number of elements currently represented in the collection. We can ignore all elements that were previously inserted and then deleted because, assuming independence of $Delete(x)$ and $Retrieve()$, this will not affect the resulting distribution. We can determine the chance of failure on $Retrieve()$ by analyzing a single probability: the probability the counter at level $\lfloor \log k \rfloor + 1$ is equal to $1$.

This is the probability that exactly one out of the $k$ elements has level $\lfloor \log k \rfloor + 1$. Since the probability a particular element has level $\lfloor \log k \rfloor + 1$ is $1/2^{\lfloor \log k \rfloor + 1}$, the probability exactly one out of the $k$ elements has level $\lfloor \log k \rfloor + 1$ is
\[
{k \choose 1}\left(\frac{1}{2^{\lfloor \log k \rfloor + 1}}\right)\left(1-\frac{1}{2^{\lfloor \log k \rfloor + 1}}\right)^{k-1} > \frac{1}{2}\left(1-\frac{1}{k}\right)^{k-1}.
\]
The right hand side approaches $1/2e$ and is above that value for all $k \geq 1$. Therefore the probability the counter at level $\lfloor \log k \rfloor + 1$ is equal to $1$ is at least $1/2e$ for all values of $k$.

We can now reason about the success probability of $Retrieve()$. We can condition on the value of $l_{max}$. Clearly, $l_{max} \geq \lfloor \log k \rfloor + 1$ with probability at least $1/2e$, since the event the counter at level $\lfloor \log k \rfloor + 1$ is equal to $1$ is included in this probability. Further, if $l_{max} > \lfloor \log k \rfloor + 1$, the probability the counter at $l_{max}$ is equal to $1$ must be at least $1/2e$, since it is less likely for an element to have a higher level. Therefore the conditional probability the counter at $l_{max}$ is equal to $1$ given that $l_{max}$ is greater than $\lfloor \log k \rfloor + 1$ is at least $1/2e$. We can then conclude that with probability at least $1/4e^2$, the counter at $l_{max}$ is equal to $1$, though with more calculation we could certainly achieve a better constant.

By maintaining $c \lceil\log n\rceil$ copies of the data structure, the probability of failure in all of them is then
\[
\left(1-\frac{1}{4e^2}\right)^{c \lceil\log n\rceil} \leq n^{c\log(1-\frac{1}{4e^2})} \leq 21n^{-c}.
\]
Therefore, for the right choice of $c$ and by union bound, we can ensure a polynomial-length sequence of $Insert(x)$, $Delete(x)$, and $Retrieve()$ has arbitrarily low probability of failure.

To analyze space complexity, observe that the space complexity is the sum of the $l_{max}$ variable for each copy of the data structure. Since each copy is independent and by linearity of expectation, the expected space complexity is $c \lceil\log n\rceil \mathbf{E}(l_{max})$. The value $l_{max}$ is the maximum of exponential variables and it can be seen that $\mathbf{E}(l_{max})$ is $O(\log n)$. Therefore the space complexity is $O(\log^2 n)$.

The time complexity of $Insert(x)$ and $Delete(x)$ is $O(1)$ + the number of additional masks/ counters inserted/ deleted. The expected number of masks/ counters inserted is no more than
\[
\sum_{i=1}^\infty i/2^i = 2
\]
therefore $Insert(x)$ and $Delete(x)$ function in constant time per copy, for $O(\log n)$ time overall.

Since the $Retrieve()$ function works with constant probability on each structure, in expectation we need only a constant number of attempts to have success. Therefore $Retrieve()$ takes $O(1)$ time. This proves Lemma \ref{montecarloL}.
\end{proof}

\section{Insert/ Delete Operations}
\label{insertdelete}

To allow for efficient insertion and deletion into $A$, instead of storing array $A$ directly, we will split $A$ into the $O(n^{1/3})$ sections between endpoints and maintain each section as its own array. When created, each section will be twice as large as necessary to provide space for insertion of elements. We will refer to the elements that reside in the extra space of each section as ``empty'' elements.

As mentioned in Section \ref{prelim}, we will associate two properties with each position in $A$: the index, which we refer to as the absolute position when all array sections are arranged in order, including the ``empty'' elements; and the rank, which, when the position is occupied by a color, denotes the number of non-empty elements that precede the position in the total order.

By $A[i]$ we refer to the $i$th index of $A$, not rank; and similarly, we let $A[i:j]$ denote the subrange of $A$ from index $i$ to index $j$. As before, $A[L:R]$ denotes the subrange of $A$ from endpoint $L$ to endpoint $R$. Though the interface of all our operations in the preceding sections function on rank, the internals of all our data structures will function on index. By storing the number of non-empty elements in each array section of $A$, we may convert from rank to index and back in $O(n^{1/3})$ time, so this distinction does not affect the running times of the previous sections.

Although the number of nonempty elements of $A$ will change with each insert/ delete operation, we will assume the variable $n$, when used to distinguish between frequent and infrequent colors, is fixed, so that designations need not change during these operations. When the number of nonempty elements of $A$ changes to require rebuilds, we may assume $n$ changes during this rebuild. In this way, the value $n$ and the number of nonempty elements of $A$ will not differ by more than a constant factor.

We will assume the version of the proposed data structure that performs updates in $O(n^{2/3})$ time. In Section \ref{givenfreq}, some of the data structures discussed require $\tilde{O}(n^{2/3})$ time per update. In this case, the insert/ delete operations take the time complexity of the update operation, that is, $\tilde{O}(n^{2/3})$ time.

\begin{lemma}
\label{insertdeleteL}
If an insertion/ deletion does not require resizing of an array segment, we may perform this operation in $O(n^{2/3})$ time.
\end{lemma}

\begin{proof}
When we insert an element at rank $i$, we determine which array it falls in by checking the sizes of each array in $O(n^{1/3})$ time. We then find the appropriate index at which it should reside in the corresponding array and push elements back in the array in $O(n^{2/3})$ time. This requires updating the corresponding entries in $D_c$ for the elements pushed back. The pointers stored in $E$ allow these updates in $O(1)$ time per element moved.

At this point, index $i$ is an open space to which we may set $A[i] = c$. We then must call Algorithm \ref{update} to update the base data structures to the addition.

Deleting an element at rank $i$ can be done similarly; pushing elements forward in the corresponding array segment, updating $D_c$ for each color $c$ that is moved, and calling Algorithm \ref{update} to update the base data structures for the removal.
\end{proof}

When an array segment has no more space for element insertions, or has less than a quarter of its elements non-empty, we need to perform a more sophisticated operation to preserve functionality and the $O(n)$ space bound. The simple solution is to just rebuild the whole data structure. Since we can do this in $O(n^{4/3})$ time, by Lemma \ref{build}, and rebuilding will only be required after $O(n^{2/3})$ insertions/ updates, we can use this approach to get amortized $O(n^{4/3} / n^{2/3}) = O(n^{2/3})$ time per insertion and deletion.

It is possible to make this bound worst-case; however, the usual blackbox trick of rebuilding as we approach the threshold for needing a new data structure does not work. One reason for this is that we need to start rebuilding for each possible segment that can over/ under flow. The rebuilt structure has size $O(n)$, so maintaining this for each segment would take $O(n^{4/3})$ space. Instead, we can try to rebuild just what we need to either break a segment into two segments, or to absorb the segment into a surrounding segment. When the number of endpoints has doubled or halved, we can rebuild the entire data structure. In this case, since we need only store a single new copy, the blackbox technique applies.

\begin{theorem}
\label{insertdeleteWC}
We may perform insert and delete in $O(n^{2/3})$ worst-case time.
\end{theorem}

\begin{proof}
Insertion/ deletion of elements without addressing resizing of a segment is handled in Lemma \ref{insertdeleteL}. When a segment either reaches full capacity or 1/4 capacity, we need to either merge the segment with a surrounding segment or split the segment into two segments. If merging the segments would result in a segment that would be above full capacity, we can instead move elements from the larger segment to the smaller segment. In any of the three scenarios, we need to achieve $O(n^{4/3})$ worst-case time and $O(n^{2/3})$ work space to be able to apply blackbox dynamization techniques to make the overall bound worst-case.

Any process of physically moving elements from one segment to another can be done similarly to as in Lemma \ref{insertdeleteL}. This takes $O(n^{2/3})$ time per move and updates $B$ arrays, $D_c$, $E$, and $F$ via Algorithm \ref{update}. Therefore, physically moving elements in any of the above three scenarios will require $O(n^{4/3})$ time and no more than $O(n^{2/3})$ work space. We now need consider the costs of adding and deleting endpoints.

When a segment is absorbed, the endpoint between the smaller segment and its neighbor will need to be removed. Doing so takes $O(\log n)$ time to modify $F$, and no more than $O(n^{1/3})$ time to delete any $B$ arrays that use this endpoint.

Similarly, when a segment is split, a new endpoint must be created. We can do this as follows. We need to populate the $B$ array for this new endpoint. There are only $O(n^{1/3})$ other end points and each $B$ array has $O(n^{1/3})$ values, so this takes $O(n^{2/3})$ work space for the structure we wish to keep. However, to create the new $B$ array, we must count each color in the corresponding ranges. If we count each color individually, via $D_c$, we can count all colors in the corresponding ranges efficiently. Since each infrequent color occurs at most $n^{1/3}$ times, and there are at most $n$ infrequent colors, this takes $O(n^{4/3})$ time and $O(n^{2/3})$ work space. 

Updating $F$ requires adding a new endpoint and counting frequent elements from the beginning of $A$ to this endpoint. We can use the counts from the preceding endpoint and add the occurrences of elements from the beginning of the segment to the position of the new endpoint. This will then take $O(n^{2/3})$ time and work space.

In Section \ref{leastfreq}, we describe a data structure that can ultimately be updated alongside $B$, but this structure may take more than $O(n^{2/3})$ space for some segments. This is okay because when the space cost of each segment is totaled, the structure takes $O(n)$ space overall. The additional data structures in Section \ref{givenfreq} can be updated alongside $B_{L, R}$.

Thus it is possible to create the new data structures we need in $O(n^{4/3})$ time and $O(n^{2/3})$ work space. To achieve worst-case insert/ delete operations, as we approach the threshold for needing to resize a segment, we rebuild the necessary structures at a rate of $O(n^{2/3})$ operations per insert/ delete in the corresponding segment. These partially-constructed data structures can be kept up to date alongside the base data structures when updates occur. If a threshold is reached in a segment, we swap the necessary structures to the global structure. One further point must also be addressed. If we create a new endpoint, then in the next iteration perform an update that creates another new endpoint, the new $B$ array for these two endpoints will not have been created. To fix this, we can build $B$ arrays to prospective endpoints as well as current endpoints during the building process. As each segment has only one prospective endpoint, this doesn't impact costs asymptotically.

Eventually, when the number of endpoints doubles or halves, we need to rebuild the whole structure to maintain that we have $O(n^{1/3})$ endpoints and $O(n^{2/3})$ elements between each endpoint. Since only one copy of the whole structure need be rebuilt, with the trick of rebuilding/ updating at twice the rate, this can be done in $O(n)$ extra space, and $O(n^{4/3} / n^{2/3}) = O(n^{2/3})$ extra time per operation.
\end{proof}

\bibliography{paper}


\end{document}